\documentclass[final]{elsarticle}
\usepackage[utf8]{inputenc}
\usepackage[english]{babel}
\usepackage{amssymb}
\usepackage[final]{changes}
\usepackage{natbib}
\usepackage{filecontents}
\usepackage{lineno,hyperref}
\usepackage{amsthm, amsmath, amsfonts}
\usepackage{booktabs, footnote}
\usepackage[utf8]{inputenc}
\usepackage[english]{babel}
\newcommand{\e}[1]{{\mathbb E}\left[ #1 \right]}
\modulolinenumbers[5]
\usepackage{floatrow}
\let\today\relax
\makeatletter
\def\ps@pprintTitle{%
	\let\@oddhead\@empty
	\let\@evenhead\@empty
	\def\@oddfoot{\footnotesize\itshape
		{Submitted to Scandinavian Actuarial Journal} \hfill\today}%
	\let\@evenfoot\@oddfoot
}
\makeatother
\newtheorem{theorem}{Theorem}

\begin{document}

\begin{frontmatter}

\title{A bridge between Local GAAP and Solvency II frameworks to quantify Capital Requirement for demographic risk}

\author[label1]{Gian Paolo Clemente}
\ead{gianpaolo.clemente@unicatt.it}
\address[label1]{Università Cattolica del Sacro Cuore, Milano}

\author[label2]{Francesco Della Corte}
\ead{francesco.dellacorte@uniroma1.it}
\address[label2]{Università La Sapienza, Roma}

\author[label1]{Nino Savelli}
\ead{nino.savelli@unicatt.it}

\begin{abstract}
The paper provides a stochastic model useful for assessing the capital requirement for demographic risk.  The model extends to the market consistent context classical methodologies developed in a local accounting framework. In particular we provide a unique formulation for different non-participating life insurance contracts and we prove analytically that the valuation of demographic profit can be significantly affected by the financial conditions in the market. A case study has been also developed considering a portfolio of life insurance contracts. Results prove the effectiveness of the model in highlighting main drivers of capital requirement evaluation, also compared to local GAAP framework.
\end{abstract}

\begin{keyword}
\texttt{life insurance, Solvency Capital Requirement, Solvency II, Local GAAP, risk theory} 
\end{keyword}

\end{frontmatter}


\section{Introduction}\label{sec:intro}
Two key innovations, brought by the directive Solvency II in insurance, are the introduction of the market consistent framework for the valuation of assets and liabilities and the definition of risk-based principles for the assessment of the Capital Requirement. In this context, the quantification of losses on an annual time horizon at a given confidence level is a crucial element in determining the requirement. Each company must decide whether to adopt the standard approach or to use its own (partial or full) internal model, which has to be approved by the local supervisory authority.  Furthermore, several sources of risk are involved in the valuation process and also measuring the dependence between them is also a crucial point.\\
In this framework, we focus on demographic profit and we provide a stochastic model to quantify the capital requirement for both mortality and longevity risk. In the literature, several papers dealt with this topic. In particular, Pitacco and Olivieri (see \cite{olivieri2008assessing}) design a framework for a market-consistent analysis of the life-annuity portfolio. The authors link the traditional approach to a risk-neutral valuation assessing the cost of capital for longevity risk and measuring the amount of target capital for mortality and longevity risk. Hári et al. (see \cite{hari2008longevity}) analyse the relevance of longevity risk for the solvency position of annuity portfolios distinguishing between micro and macro-longevity risks.\\
Stevens et al.  (see \cite{stevens2010calculating}) quantify the value of annuity liabilities and of the related longevity risk capital requirement by applying the classical Lee-Carter model to estimate the uncertainty of future survival probabilities. Bauer et al. (see \cite{bauer2015least}) propose an approach for the calculation of the required risk capital based on least-squares regression and Monte Carlo simulations. Savelli and Clemente (see \cite{savelli2013risk}) provide an approach based on risk theory to evaluate the capital requirement for mortality and longevity risk in a local accounting framework. Boonen (see \cite{boonen2017solvency}) examines the consequences for a life insurance company of a calibration of longevity and financial risks by using the expected shortfall instead of value-at-risk. Furthermore, the notion of market-consistent valuation of insurance liabilities has been investigated recently by several authors (see \cite{pelsser2014time}, \cite{dhaene2017fair}) and the analyzes have been also extended to a dynamic multi-period setting (see \cite{barigou2019fair}).
\added{Dahl et al. (see \cite{dahl2004stochastic}) model mortality intensity as a stochastic process and  quantify mortality risk by capturing the importance of time dependency and uncertainty.}
This paper contributes to the existing literature proposing a different approach for modelling the capital requirement for mortality and longevity risk. We adapt the well-known gain and loss decomposition (see \cite{bacinello1986portfolio}) to a Solvency II framework, providing a closed formula for the random variable demographic profit and loss of a life insurance company. We show that this general formula holds for different \added{non-participating} life insurance contracts and we analytically split it in several elements in order to emphasize \added{the} main key drivers. \added{To the best of our knowledge a closed formula is not already available in the literature in a market-consistent framework. Additionally, it is noteworthy that, although our aim is not to forecast future mortality rates, the proposed approach is also consistent with the application of mortality models.}\\
In particular, we assume a portfolio of \added{non-participating} life insurance policies composed by several cohorts of contracts, where in each cohort all the policyholders have the same characteristics (e.g., age, gender etc.) and the only element of distinction is represented by the \added{sums insured}.\\
Through this model, we prove the analytical decomposition of the expected demographic profit/loss highlighting main drivers. Additionally, we assess the distribution via Monte Carlo simulations and we identify a Solvency Capital Requirement compliant with Solvency II. Furthermore, we assure that the model accurately reflects different sources of profit as requested by the requirements set by the Solvency II Directive (see art. 123, \cite{delacts}). In addition, a numerical section is presented to apply the proposed approach on alternative life insurance policies; \added{here, we focus on the risk identified in Shen et al. (see \cite{shen2018lifetime}) as idiosyncratic risk}.\\
The remainder of this paper is organized as follows. Section 2 exploits the traditional Homans formula, based on local Generally Accepted Accounting Principles (GAAP), adapting it to a market consistent context. Section 3 focuses on the random variable demographic profit, also providing a decomposition of this random variable useful for profit and loss composition. \added{Section 4 presents the algebra underlying the model and the proof of a recursive formula that supports the results.} Section 5 presents some simplified examples to emphasize the main key drivers. In Section 6, we assess our proposal by developing detailed case studies based on a real portfolio composed by different \added{non-participating} life insurance contracts. In particular, main results confirm the effectiveness of the model in measuring the capital requirement and assessing the different components that affect the demographic results. In Section 7 we propose the conclusions of our paper.

\section{Technical Profit and gain/loss decomposition in a Solvency II framework}

\noindent In this Section, we provide a formula for the random variable (r.v.) technical profit of a life insurance company. In particular, we define it in a market consistent framework according to the definition of technical liabilities given by the Solvency II regulation. Furthermore, we prove that a gain and loss decomposition is possible in order to emphasize main profit components. The decomposition is developed in a similar fashion of the well-known decomposition provided by Homans (see \cite{olivieri2005valutazione}) in a local accounting framework.\\
We start by briefly recalling the r.v. technical profit, that we denote\footnote{from now on, each random variable will be indicated with the tilde} with $\tilde{Y}_{t+1}$, in a local accounting context (see \cite{savelli1993modello})
\begin{equation}
\tilde{Y}_{t+1} = [\tilde{VB}_t + \tilde{B}_{t+1} -\tilde{E}_{t+1} -\tilde{S}_{t+1} ](1+\tilde{j}_{t+1})-[\tilde{X}_{t+1}+\tilde{VB}_{t+1}]
\label{eq:tplocalgaap}
\end{equation}
The random variable technical profit $\tilde{Y}_{t+1}$ is defined as the difference of two terms. The first one is the sum of the complete technical provisions $\tilde{VB}_{t}$\footnote{As "complete", $\tilde{VB}_{t}$ indicates the sum of the pure mathematical reserve (expected present value of the benefits net of the expected present value of the premiums) and the expense reserve. Both are calculated on locked and prudential (demographic and financial) bases: hence both demographic and financial base used is the same as applied in the pricing phase.}, and the gross earned premiums $\tilde{B}_{t+1}$ net of expenses $\tilde{E}_{t+1}$ and surrenders $\tilde{S}_{t+1}$ \added{accumulated} at the \deleted{effective}\added{actual} financial return rate $\tilde{j}_{t+1}$. The second term in formula (\ref*{eq:tplocalgaap}) is the sum of total claim costs $\tilde{X}_{t+1}$ and the complete technical provisions stored at the end of the year $\tilde{VB}_{t+1}$. It is noteworthy that we consider with the r.v. $\tilde{X}_{t+1}$ different \added{non-participating} life insurance policies (i.e., as specified later, this r.v. could be equal to zero or to the total lump sums paid in case of death or survival according to the kind of benefit covered by the insurance policy).\\
With the introduction of the Solvency II framework (see Directive 138/2009/EC), previous formula has to be adapted to consider the market consistent valuation of assets and liabilities. Referring only to non-hedgeable liabilities, these are calculated as the sum of Best Estimate and Risk Margin (see Art. 77 of Directive 2009/138/EC). Hence, we can rewrite formula (\ref*{eq:tplocalgaap}) as:
\begin{equation}
\tilde{Y}_{t+1}^{MCV} = [\tilde{BE}_t + \tilde{B}_{t+1} -\tilde{E}_{t+1} -\tilde{S}_{t+1} ](1+\tilde{j}_{t+1})-[\tilde{X}_{t+1}+\tilde{BE}_{t+1}]
\label{eq:tpmcv}
\end{equation}
It should be pointed out that the Risk Margin is not considered \added{since the purpose of this model is to identify a Solvency Capital Requirement consistent with the legislation and, therefore, with the Delegated Acts}. Indeed, Delegated Regulation (see \cite{delacts}) assumes that worst case scenario does not change the amount of Risk Margin included in the technical provisions. \added{Therefore, the inclusion of this \added{aforementioned} component in formula (\ref{eq:tpmcv}) when the formula is used for \added{risk} capital purposes is not in line with the constraints introduced by the regulation and could inappropiately affect the comparison with the standard formula \deleted{capital} requirement.}
Similarly to the well-known results regarding the decomposition of the classical Homans formula, we prove in the \ref*{sec:app1} the five components of the technical profit formula \added{in a stochastic context}. Our main interest is to focus on the first and the fourth components, respectively the demographic and financial profit and loss.
Using rate-based expressions instead of amounts (therefore switching from uppercase letters to lowercase ones), the demographic profit is defined as:
\textbf{\begin{equation}
\begin{aligned}
_1\tilde{y}_{t+1}^{MCV} = \ & \{[\tilde{be}_t^{\tilde{Rf(t)}, q(t)} + b_{t+1} \cdot (1-\alpha^*-\beta^*)-\gamma^*]\cdot(\tilde{w}_t -\tilde{s}_{t+1})\cdot(1+j^*)+ \\ &-(\tilde{x}_{t+1}+\tilde{w}_{t+1}\cdot\tilde{be}_{t+1}^{\tilde{Rf(t+1)}, q(t+1)})
\} \cdot 
\end{aligned}
\label{eq:demographicprofit}
\end{equation}}
\noindent where $\tilde{be}_t^{Rf(t), q(t)}$ is the rate of best estimate based on the risk-free curve available at time t and realistic demographic assumptions $q(t)$ at time $t$, $b_{t+1}$ is the premium rate, $\alpha^*$, $\beta^*$ and $\gamma^*$ are expense loading coefficients for acquisition, management and collection costs, respectively. The term $(\tilde{w}_t -\tilde{s}_{t+1})$ is the difference between the amount of the \added{sums insured}  $\tilde{w}_t$ at the end of time $t$ and surrenders $\tilde{s}_{t+1}$; $j^*$ is the first order financial rate (i.e., the technical rate assumed in the premium assessment). Similarly the reserve at the end of the year is computed as the product of \added{sums insured} $\tilde{w}_{t+1}$ and the best estimate rate $\tilde{be}_t^{Rf(t+1), q_{t+1}}$  based on financial $Rf(t+1)$ and demographic $q(t+1)$ assumptions in force at time $t+1$. \added{It is noteworthy that formula (\ref{eq:demographicprofit}) also allows to analyse the effects of a possible change in the demographic technical base in $t+1$. Additionally, it is easy to show that exploiting the term $\tilde{x}_{t+1}$ it is possible to find a capital requirement framework in line with the proposal made in the Quantitative Impact Study \added{n.2}, carried out in may 2006 for the preparation of final SCR standard formula, where a closed notation was provided for both mortality and longevity risk.}\\
\noindent A peculiar mention must be given also to \added{another} profit component, the so-called financial profit, that is necessary to explain some key aspects of the model. It is defined as follow:
\textbf{\begin{equation}
	\begin{aligned}
	_2\tilde{y}_{t+1}^{MCV} = \ &(\tilde{j}_{t+1}-j^*)\cdot[\tilde{be}_t^{\tilde{Rf(t)}, q(t)}\cdot\tilde{w}_t+b_{t+1}(1-\alpha^*-\beta^*)\cdot(\tilde{w}_t - \tilde{s}_{t+1})+\\
	&-\gamma^*\cdot\tilde{w}_t-g^*_t\cdot\tilde{be}_t^{\tilde{Rf(t)},q(t)}\cdot\tilde{s}_{t+1}]
	\end{aligned}
	\label{eq:financialprofit}
	\end{equation}}
where $\tilde{j}_{t+1}$ is the r.v. that describes the rate of return of invested assets and $g^*_t$ is a specific penalization coefficient applied in case of surrender (see \ref*{sec:app1} for details on the formula).
As expected, the sign of this component depends on the relation between the effective investment rate and the technical rate guaranteed to the policyholders. 
\added{For the sake of brevity, the other profit components are reported in Appendix A, in particular that for expenses, for lapses and a residual margin. As proved, the sum of these five components gives back the whole technical profit (see formula (\ref*{eq:tpmcv}))}

\section{The demographic profit and its factorisation}
\noindent Given the gain and loss decomposition provided in previous Section, we focus now only on the demographic component (see formula (\ref{eq:demographicprofit})). Indeed, modelling this random variable, we are able to assess the capital requirement for mortality or longevity risk. 
After some simple manipulations, it is possible to rewrite formula (\ref{eq:demographicprofit}) as follows:
\begin{equation}
\begin{aligned}
_1\tilde{y}_{t+1}^{MCV} = \ &D_{t+1}^b \cdot[q_{x+t}^*\cdot(\tilde{w}_t-\tilde{s}_{t+1})-\tilde{z}_{t+1}]+ \\ 
&+(\tilde{be}_t^{\tilde{Rf}_t, q}-v_t^b)\cdot(\tilde{w}_t -\tilde{s}_{t+1})\cdot(1+j^*)-\tilde{w}_{t+1}\cdot(\tilde{be}_{t+1}^{\tilde{Rf}_{t+1},q}-v_{t+1}^b)
\end{aligned}
\label{eq:dp2}
\end{equation}
where $D_{t+1}^b$ is the complete sum-at-risk rate at time $t+1$ and $q^*$ is the first-order annual death probability, $v_t^b$ is the complete reserve rate based on first order basis. \added{The sum-at-risk are here defined \lq\lq complete\rq\rq since complete reserve rate is considered (i.e. including expenses reserves); they can be negative, as in pure endowment and annuities cases.} This formula is based on the following relation that describes the evolution of the \added{sums insured} over time:
\begin{equation}
\tilde{w}_{t+1}=\tilde{w}_t-\tilde{s}_{t+1}-\tilde{z}_{t+1}
\label{eq:recursivesums}
\end{equation}
where $\tilde{z}_{t+1}$ is the amount of \added{sums insured}  eliminated in case of death.\footnote{In this context $\tilde{x}_{t+1}=\tilde{z}_t+1$ for term Insurance and Endowment Policies; $\tilde{x}_{t+1}=0$ for Pure Endowment policies} 
In formula (\ref{eq:dp2}), we assume that the best estimates in $t$ and in $t+1$ are calculated on the same realistic demographic assumptions. Therefore, we simplify the notation using $q$ instead of  $q(t)$ and $q(t+1)$. Since the assessment is carried out on a one-year time span, it makes sense to consider that within such a short period of time the insurer does not change its demographic expectations. However, it is possible to evaluate the additional effect on the demographic profit of a one-year change of the second-order demographic assumptions. \\
It is worth pointing out that the first term in formula (\ref{eq:dp2}), here denoted as $_1\tilde{y}_{t+1}^{LG}$:
\begin{equation}
_1\tilde{y}_{t+1}^{LG} = \ D_{t+1}^b \cdot[q_{x+t}^*\cdot(\tilde{w}_t-\tilde{s}_{t+1})-\tilde{z}_{t+1}]
\label{localgaapeq}
\end{equation}
represents the demographic profit in a local accounting framework (see \cite{savelli2013risk}).\\
Additionally, it is interesting to separately highlight the effects of risk-free rates volatility and demographic trends. In this regard, we provide the following decomposition: 
\begin{equation}
_1\tilde{y}_{t+1}^{MCV} = _1\tilde{y}_{t+1}^{LG} +_1\tilde{y}_{t+1}^{MCV Rf-j*} +_1\tilde{y}_{t+1}^{MCV q-q*} 
\label{demprof3comp}
\end{equation}
where  $_1\tilde{y}_{t+1}^{MCV Rf-j*}$  is the demographic profit given by the differences between the first order financial rate $j^*$  and the risk-free rate curve. The term  $_1\tilde{y}_{t+1}^{MCV q-q*}$ measures instead the demographic profit originated by the differences between first-order and second-order death probabilities ($q^*$ and $q$, respectively).
To provide this decomposition and to define the second and the third term in formula (\ref{demprof3comp}), we denote with $epv_t^{j^*,q}$ the expected present value of future cash-flows evaluated at time $t$  and computed by using first-order discount rates $j^*$ and second-order death probabilities $q$. In other words, this amount differs from $\tilde{be}_t^{Rf(t),q}$ only in terms of discounting factors.\\
It is easy to show that  $_1\tilde{y}_{t+1}^{MCV Rf-j*}$  is defined as:
\begin{equation}
\begin{aligned}
_1\tilde{y}_{t+1}^{MCV Rf-j*}=&(\tilde{w}_t-\tilde{s}_t)\cdot[(\tilde{be}_t^{Rf(t),q}-epv_t^{j^*,q})(1+j^*)-(\tilde{be}_{t+1}^{Rf(t+1),q}-epv_{t+1}^{j^*,q})]+\\
&(\tilde{be}_{t+1}^{Rf(t+1),q}-epv_{t+1}^{j^*,q})\cdot\tilde{z}_{t+1}
\end{aligned}
\label{2ndComponent}
\end{equation}
and the third component in formula (\ref{demprof3comp}) is:
\begin{equation}
\begin{aligned}
_1\tilde{y}_{t+1}^{MCV q-q*}=&(\tilde{w}_t-\tilde{s}_t)\cdot[(epv_t^{j^*,q}-v_t^b)(1+j^*)-(epv_{t+1}^{j^*,q}-v_{t+1}^b)]+\\
&(epv_{t+1}^{j^*,q}-v_{t+1}^b)\cdot\tilde{z}_{t+1}
\end{aligned}
\label{3rdComponent}
\end{equation}
It is interesting to note that $_1\tilde{y}_{t+1}^{MCV Rf-j*}$  is strictly related to the difference between the best estimate and the expected present value of future cash-flows $epv_t^{j^*,q}$. In particular, we have a positive value if this difference, \added{accumulated} at the technical rate $j^*$, is greater than the analogous difference computed at time $t+1$. 
It's essential to note that the difference $(\tilde{be}_t^{Rf(t),q}-epv_t^{j^*,q})$ only depends on the difference between the risk-free rates at time t and the technical rate $j^*$. For instance, a sudden and substantial change in value of the risk-free curve entails a jump in $t+1$ greater than the jump in $t$.\\
We also recall that we are focusing here only on the demographic profit. Hence, the whole effect of the risk-free rate curve on the technical profit can be assessed by considering also the financial profit. However, this point goes beyond the scope of the paper that is to quantify the capital requirement for mortality or longevity risk.\\
The term $_1\tilde{y}_{t+1}^{MCV q-q*}$  depends instead on the difference between the expected present value  $epv_t^{j^*,q}$   and the technical provisions defined according to local accounting rules. Both values are computed with the same discounting rate $j^*$, but using a different life table (second-order and first-order, respectively). \added{It \added{must be pointed out} that the sign of $_1\tilde{y}_{t+1}^{MCV q-q*}$ is mainly related to the comparison between the term $(epv_t^{j^*,q}-v_t^b)$, accumulated at the technical rate $j^*$, and the same difference evaluated at time $t+1$. We have indeed that the last term in formula (\ref{3rdComponent}) has usually a very low weight and hence a low effect on the sign of $_1\tilde{y}_{t+1}^{MCV q-q*}$.}\\
Now, we focus on the expected demographic profit (i.e. $_1\tilde{y}_{t+1}^{MCV}$) and we prove (see next Section) that when $t>1$, we have: 
\begin{equation}
\e{_1\tilde{y}_{t+1}^{MCV}}=\e{_1\tilde{y}_{t+1}^{MCV Rf-j*}}
\label{EV1}
\end{equation}
since
\begin{equation}
\e{_1\tilde{y}_{t+1}^{LG}}=-\e{_1\tilde{y}_{t+1}^{MCV q-q*}}
\label{EV2}
\end{equation}
This result shows that in case the second-order life table chosen by the insurance \added{company} is stable in the period $(t,t+1]$, the sign of the average demographic profit depends mainly on the differences between the risk-free curve and first-order financial rate. In particular, by formula (\ref{2ndComponent}), we can also rewrite the mean as (when $t\geq1$):
\begin{equation}
\begin{aligned}
&\e{_1\tilde{y}_{t+1}^{MCV}}=\e{_1\tilde{y}_{t+1}^{MCV Rf-j*}}\\
&\e{_1\tilde{y}_{t+1}^{MCV}}=\e{\tilde{w}_t}\cdot(1+j^*)\cdot\Biggl[b_t+\e{\tilde{be}_t^{Rf(t),q}}-\e{\tilde{be}_{t+1}^{Rf(t),q}} \cdot _{/1}E_{x+t}^{j^*,q}-\e{\tilde{x}_{t+1}}\Biggr]
\end{aligned}
\label{EVtgreater1}
\end{equation}
\added{On the bases of} formula (\ref{EVtgreater1}) it easy to show that (see Section 4):
\begin{equation}
	\e{_1\tilde{y}_{t+1}^{MCV}}=0
	\label{ev0}
\end{equation}
if, at time $t$, the one-year spot rate is equal to $j^*$.\\
It is worth pointing out that, at the inception of the contract (i.e. when $t=0$), by formula (\ref{EVtgreater1}) the expected value of the demographic profit also depends on the differences between first-order and second-order life tables.

\section{Model algebra and underlying recursive formula}
\noindent We report in this section main results and related proofs needed to support the presented stochastic model. 
\begin{theorem}
	Considering a without-profit life insurance contract, we have that $\e{_1\tilde{y}_{t+1}^{LG}}=-\e{_1\tilde{y}_{t+1}^{MCV q-q*}}$ for $t>1$.	
	\label{teorema_1}
\end{theorem}
\begin{proof}
\noindent The starting point is represented by the sum of formulas (\ref{localgaapeq}) and (\ref{3rdComponent}). However, we specify that, instead of using the compact notation of formula (\ref{localgaapeq}), we consider the extended definition of demographic profit:
\begin{equation*}
	\begin{aligned}
		_1\tilde{y}_{t+1}^{LG}&=[v_t^b+\pi]\cdot(\tilde{w}_t-\tilde{s}_{t+1})\cdot(1+j^*)-(\tilde{x}_{t+1}+\tilde{w}_{t+1}\cdot v_{t+1}^b)\\
		&=\ D_{t+1}^b \cdot[q_{x+t}^*\cdot(\tilde{w}_t-\tilde{s}_{t+1})-\tilde{z}_{t+1}]
	\end{aligned}
	\label{eq:proof_ev_0}
\end{equation*}
where $\pi$ is the constant pure premium rate obtained as $\pi=\pi_{t+1}=b_{t+1} (1-\alpha^*-\beta^* )-\gamma^*$.
The proof regarding the equality between the two equations can be found in \cite{savelli2013risk}. Hence, we have:
\begin{equation*}
	\begin{aligned}
		_1\tilde{y}_{t+1}^{LG}+_1\tilde{y}_{t+1}^{MCV q-q^*}=&\ v_t^b\cdot(\tilde{w}_t -\tilde{s}_{t+1})\cdot(1+j^*)+\pi\cdot (\tilde{w}_t -\tilde{s}_{t+1})\cdot(1+j^*) \\
		-v_{t+1}^b \cdot \tilde{w}_{t+1}&-\tilde{x}_{t+1}+epv_{t}^{j^*, q}\cdot(\tilde{w}_t -\tilde{s}_{t+1})\cdot(1+j^*) \\
		-v_t^b\cdot (\tilde{w}_t -&\tilde{s}_{t+1})\cdot(1+j^*)-(epv_{t+1}^{j^*, q}-v_{t+1}^b)\cdot(\tilde{w}_t -\tilde{s}_{t+1}-\tilde{z}_{t+1})
	\end{aligned}
	\label{eq:proof_ev_1}
\end{equation*}

\noindent By formula (\ref{eq:recursivesums}),  it is possible to rewrite previous relation as follows:
\begin{equation*}
	\begin{aligned}
		_1\tilde{y}_{t+1}^{LG}+_1\tilde{y}_{t+1}^{MCV q-q^*}=&\left(\pi+epv_{t}^{j^*, q}\right)\cdot (\tilde{w}_t -\tilde{s}_{t+1})\cdot(1+j^*) \\ &-\tilde{x}_{t+1}-epv_{t+1}^{j^*, q}\cdot(\tilde{w}_t -\tilde{s}_{t+1}-\tilde{z}_{t+1})\\
		=(\tilde{w}_t -&\tilde{s}_{t+1})\cdot\left[\left(\pi+epv_{t}^{j^*, q}\right)\cdot(1+j^*)-epv_{t+1}^{j^*, q}\right]\\
		+ epv_{t+1}^{j^*, q}&\cdot\tilde{z}_{t+1}-\tilde{x}_{t+1}
	\end{aligned}
\end{equation*}

\noindent We compute now the expected value and we consider the case of an endowment policy. The result is easily extendable to pure endowment and term insurance policies, since these contracts can be seen as specific cases of an endowment contract.
\begin{equation*}
	\begin{aligned}
		\e{_1\tilde{y}_{t+1}^{LG}+_1\tilde{y}_{t+1}^{MCV q-q^*}}=&\\
		\e{\tilde{w}_t-\tilde{s}_t}\left[\left(\pi+epv_{t}^{j^*, q}\right)\cdot(1+j^*)-epv_{t+1}^{j^*, q}\right]
		&+epv_{t+1}^{j^*,q}\cdot\e{\tilde{z}_{t+1}}-\e{\tilde{x}_{t+1}}
	\end{aligned}
\end{equation*}
Since $\tilde{x}_{t+1}=\tilde{z}_{t+1}$  for an endowment contract and $\e{\tilde{z}_{t+1}}=q_{x+t}\cdot\e{\tilde{w}_t-\tilde{s}_t}$, we have:
\begin{equation*}
	\begin{aligned}
		\e{_1\tilde{y}_{t+1}^{LG}+_1\tilde{y}_{t+1}^{MCV q-q^*}}=& \\
		\e{\tilde{w}_t-\tilde{s}_t}\Biggl(\left(\pi+epv_t^{j^*,q}\right)\cdot(1+j^*)&
		-epv_{t+1}^{j^*,q}\cdot p_{x+t}-q_{x+t}\Biggr)
	\end{aligned}
\end{equation*}

Focusing on the term in the bracket, we prove that the following recursive equation holds:
\begin{equation}
	\left(\pi+epv_t^{j^*,q}\right)\cdot(1+j^*)=epv_{t+1}^{j^*,q}\cdot p_{x+t}+q_{x+t}
	\label{equazione_30}
\end{equation}
In the case of a endowment contract with unitary sum insured, we have:
\begin{equation*}
	\begin{aligned}
		epv_t^{j^*,q}=&_{n-t}p_{x+t}\cdot(1+j^*)^{-(n-t)}+\sum_{h=0}^{n-t-1} {_{h/1}{q_{x+t}}\cdot(1+j^*)^{-(h+1)}}+\\
		&-\pi \sum_{h=0}^{n-t-1}{_{h}{p_{x+t}}\cdot(1+j^*)^{-h}} = \\
		&_ {n-t}p_{x+t}\cdot(1+j^*)^{-(n-t)}+\\
		&+(q_{x+t}\cdot(1+j^*)^{-1}+\sum_{h=1}^{n-t-1}{_{h/1}q_{x+t}\cdot (1+j^*)^{-(h+1)}})+\\
		&-\pi\cdot(1+\sum_{h=1}^{n-t-1}{_hp_{x+t}\cdot(1+j^*)^{-h}})
	\end{aligned}
\end{equation*}
That, for $s=h-1$, could be rewritten as:
\begin{equation*}
	\begin{aligned}
		epv_t^{j^*,q}=&_{n-t}p_{x+t}\cdot(1+j^*)^{-(n-t)}+\\
		&+(q_{x+t}\cdot(1+j^*)^{-1}+\sum_{s=0}^{n-t-2}{_{(s+1)/1}q_{x+t}\cdot(1+j^*)^{-(s+2)}})+\\
		&-\pi\cdot(1+\sum_{s=0}^{n-t-2}{_sp_{x+t}\cdot(1+j^*)^{-(s+1)}})
	\end{aligned}
\end{equation*}
Considering now:
\begin{equation*}
	\begin{aligned}
		\dfrac{epv_t^{j^*,q}}{_1E_{x+t}}=&_{n-t-1}p_{x+t+1}\cdot(1+j^*)^{-(n-t-1)}+\\
		&\left(\dfrac{q_{x+t}}{_1p_{x+t}}+\sum_{s=0}^{n-t-2}{\dfrac{_{(s+1)/1}q_{x+t}}{_1p_{x+t}}\cdot(1+j^*)^{-(s+1)}}\right)+\\
		&-\dfrac{\pi}{_1E_{x+t}}-\pi\cdot\sum_{s=0}^{n-t-2}{_sp_{x+t+1}\cdot(1+j^*)^{-s}}
	\end{aligned}
\end{equation*}
we have:
\begin{equation*}
	\dfrac{epv_t^{j^*,q}}{_1E_{x+t}}=epv_{t+1}^{j^*,q}-\dfrac{\pi}{_1E_{x+t}}+\dfrac{q_{x+t}}{_1p_{x+t}}
\end{equation*}
and, with simple algebra, easily follows equation (\ref*{equazione_30}):
\end{proof}
\newpage

\noindent Having proved formula (\ref{EV1}), the final purpose of this Section is to consider the demographic risk expressed in formula (\ref{eq:demographicprofit}) and prove the following theorem:
\begin{theorem}
	Considering a without-profit endowment insurance contract that pays a lump sum equal to 1 either in case of death or in case of survival at the end of the contract and without benefits in case of lapses, if second-order technical bases at time $t$ and time $t+1$ are the same, the following recursive equation holds:
	\begin{equation}
		(be_t+\pi)(1+i_{t}(0,0,1))=_{/1}q_{x+t}+be_{t+1}\cdot_1p_{x+t}
	\end{equation}
where $i_{t}(0,1)$ is the spot rate in force at time $t$ and  $be_t$ and $be_{t+1}$ are the pure best estimate rates computed using realistic demographic and financial assumptions in force at time t and neglecting expenses and expenses loadings.
	\label{teorema_2}
\end{theorem}
\begin{proof}
We recall here the definition of the best estimate rate of an endowment policy computed using realistic demographic assumption $q$ and the risk-free rate curve $it_{}$ in force at time $t$. 
\begin{equation*}
	\begin{aligned}
		be_t=&_{n-t}p_{x+t}\cdot\Biggl[\prod_{h=0}^{n-t-1}{(1+i_t(0,h,h+1))}\Bigg]^{-1}+\\
		&\sum_{k=0}^{n-t-1}{_{k/1}q_{x+t}\Biggl[\prod_{h=0}^{k}{(1+i_t(0,h,h+1))}\Biggr]^{-1}}-\pi\cdot \ddot{a}_{(x+t):(n-t)}
	\end{aligned}
\end{equation*}
This is also equal to:
\begin{equation}
	\begin{aligned}
		be_t=&_{n-t}p_{x+t}\Biggl[\prod_{h=0}^{n-t-1}{(1+i_t(0,h,h+1))}\Biggr]^{-1}+_{/1}q_{x+t}\cdot(1+i_t(0,0,1))^{-1}+\\
		&\sum_{k=1}^{n-t-1}{_{k/1}q_{x+t}\Biggl[\prod_{h=0}^{k}{(1+i_t(0,h,h+1))}\Biggr]^{-1}}-\pi\cdot\sum_{h=1}^{n-t-1}{_hE_{x+t}}-\pi
	\end{aligned}
	\label{eq:trepunto22}
\end{equation}
From formula (\ref{eq:trepunto22}), we have that the following relation holds:
\begin{equation}
	\begin{aligned}
		(be_t+\pi)\cdot(1+i_t(0,0,1))&-_{/1}q_{x+t}=_{n-t}p_{x+t}\Biggl[\prod_{h=1}^{n-t-1}{(1+i_t(0,h,h+1))}\Biggr]^{-1}\\
		&+\sum_{s=0}^{n-t-2}{\Biggl( { }_{(s+1)/1} q_{x+t}\Biggl[\prod_{h=1}^{s+1}{(1+i_t(0,h,h+1))}\Biggr]^{-1}\Biggr)}\\
		&-\pi\cdot\sum_{h=1}^{n-t-1}{\Biggl({ }_h p_{x+t}\Biggl[\prod_{j=1}^{h}{(1+i_t(0,j,j+1))}\Biggr]^{-1}\Biggr)}
	\end{aligned}
	\label{eq:trepunto33}
\end{equation}
Since, the estimation of the best estimate at time $t+1$ under the assumption in force at time $t$ is equal to:
\begin{equation}
	\begin{aligned}
		be_{t+1}=&_{n-t-1}p_{x+t+1}\Biggl[\prod_{h=1}^{n-t-1}{(1+i_t(0,h,h+1))}\Biggl]^{-1}+\\
		&+\sum_{k=0}^{n-t-2}{\Biggl( { }_{k/1}q_{x+t+1}\cdot\Biggl[\prod_{h=1}^{k+1}{(1+i_t(0,h,h+1))}\Biggr]^{-1}\Biggr)}+\\
		&-\pi\cdot\ddot{a}_{(x+t+1):(n-t-1)}
	\end{aligned}
	\label{eq:trepunto44}
\end{equation}
it is noticeable that the right-hand side of formula (\ref{eq:trepunto33}) is equal to $be_{t+1}\cdot _1p_{x+t}$. Hence, we have:
\begin{equation}
	(be_t+\pi)(1+i_t(0,0,1))=_{/1}q_{x+t}+be_{t+1}\cdot_1p_{x+t}
	\label{finale_ricorsiva}
\end{equation}

\noindent It is worth pointing out that the proof can be easily adapted to the cases of single premiums, pure endowment or term insurance contracts and flat rates, that have been also analysed in the paper. All of these combinations can be considered as special cases of the one that has been proved.
\end{proof}

\added{\noindent The most interesting aspect of this recursive formula, concerns the fact that formula (\ref{eq:demographicprofit}) for $t>1$, has an expected value different from $0$ when the technical rate $j^*$ differs from the spot rate $i_t(0,1)$. This difference between formula (\ref{eq:demographicprofit}) and formula (\ref*{finale_ricorsiva}) will be relevant for the interpretation of the results reported in Sections 5 and 6.}

\section{The profit formation} 
\noindent We focus in this Section on the expected demographic profit considering two \added{non-participating} life insurance contracts: a pure endowment and a term insurance. To this end, we consider a cohort of policyholders, whose main characteristics are summarized in Table \ref{tab:tabella1}. We start from a simplified application that allows to provide additional insights. In particular, for the sake of simplicity, we assume flat risk free-rates constant over time and we neglect the effect of expenses. Additionally, we are assuming that the risk-free rates \added{is equal to} the technical rate $j^*$ guaranteed to the policyholder. We have instead that the insurance company priced contracts assuming first-order death probabilities $q^*$ equal to $85\%$ of the death probabilities given by the ISTAT2016 \added{population} life table.  
In this first analysis, we assume that the observed mortality follows rates given by the ISTAT2016 \added{population} life table. In other words, a prudential pricing has been applied by the insurance company and hence, a demographic profit is expected. In this regard, we compare in Figure \ref{fig:F1} how the profit is released over time in either a market-consistent or a local accounting framework. 
	\begin{table}[htb]
		\centering
		\caption{Model parameters}
		\begin{tabular}{|l|r|}
			\hline
			Individual age at policy issue & 40             \\ \hline
			Gender                         & Male           \\ \hline
			Policy duration                & 20             \\ \hline
			Expenses loadings              & 0\%            \\ \hline
			Risk-free rates                & 1\%            \\ \hline
			Number of policyholders        & 15,000         \\ \hline
			Initial \added{sums insured} & 1,510,653,999        \\ \hline
			CV\footnote{\added{CV stands for Coefficient of Variation}} of initial \added{sums insured}                   & 1.99              \\ \hline
			I order demographic basis       & 0.85*ISTAT2016 \\ \hline
			II order demographic basis      & ISTAT2016      \\ \hline
			Technical rate         & 1\%            \\ \hline
		\end{tabular}
	\label{tab:tabella1}
	\end{table}
In a local accounting framework (i.e. see $_1\tilde{y}_{t+1}^{LG}$, formula (\ref{localgaapeq})), the expected profit varies over time, depending on the trend of the sum-at-risk, the implicit safety loading and the effects of mortality.
As expected, in a market consistent context, because of unlocked technical bases, the expected profit (i.e. $_1\tilde{y}_{t+1}^{MCV}$) occurs when the difference between the technical bases and realistic assumptions is revealed, while only the unexpected profit linked to the \added{idiosyncratic} risk of the mortality rates occurs over time. In particular, at the inception of the contract, we notice the effect of the difference between first-order basis (used for premium assessment) and second-order basis (used for the market-consistent valuation of technical provisions). In other words, implicit safety loadings are released as a technical profit at the end of the first year.\\
From year 2 to the end of the contract, we observe that $\e{_1\tilde{y}_{t+1}^{MCV}}$ is equal to 0. We are indeed assuming that risk-free rates are constant and equal to first-order technical rate $j^*$ and that the realistic assumptions used for best estimate valuation are kept unchanged over time by the company: \added{these results derive from Theorem (\ref{teorema_1}) and Theorem (\ref*{teorema_2}) of Section 4.}\\
\begin{figure}[ht!]
	\centering
	\includegraphics[width=110mm]{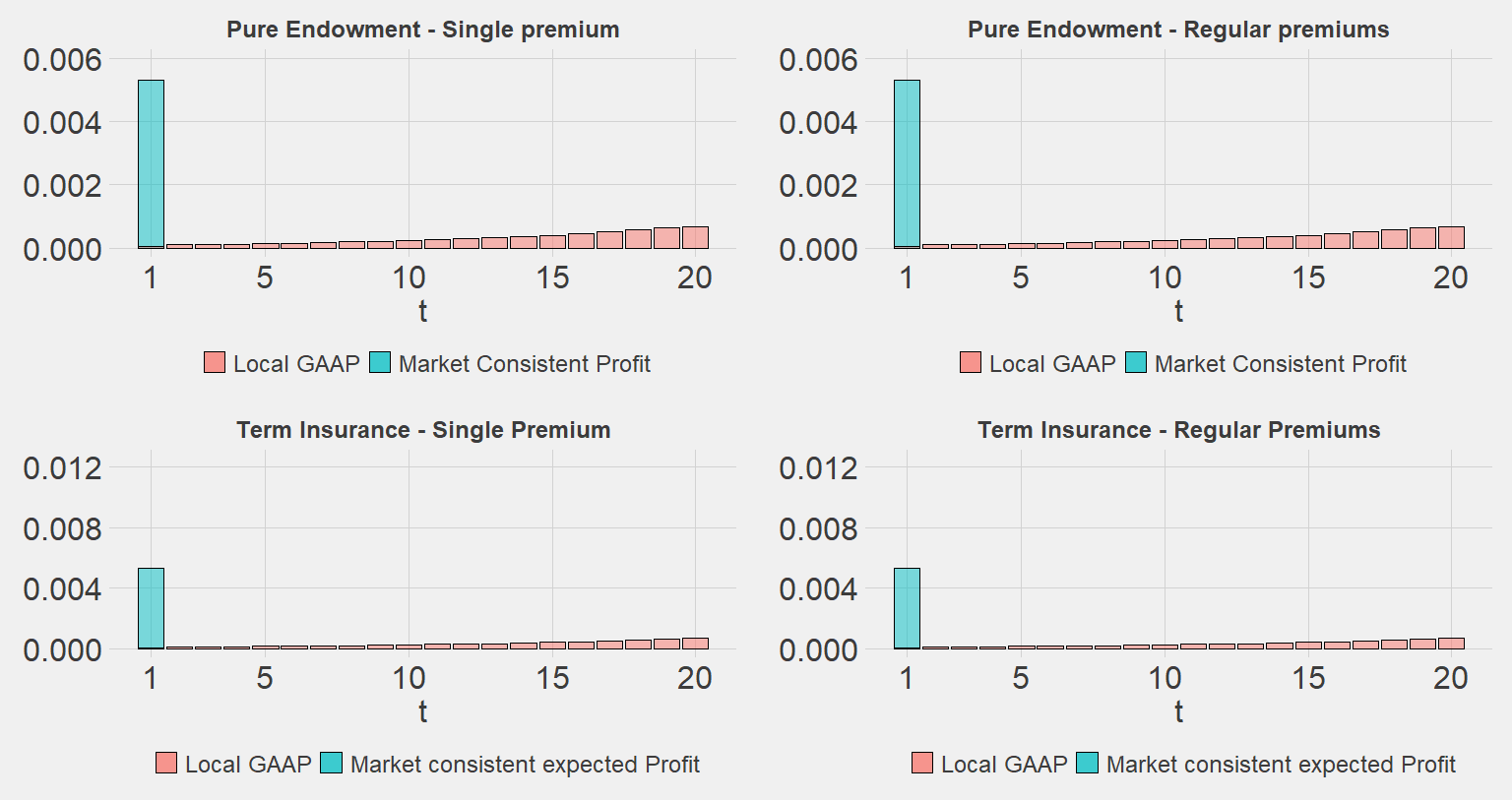}
	\caption{Pattern of expected demographic profit for a Pure Endowment and a Term Insurance \label{overflow}}
	\label{fig:F1}
\end{figure}
\noindent Same analysis has been also applied to term insurance contracts. In this case, we maintain the same assumptions reported in Table 1, but we assume that premiums have been computed using a ISTAT2014 life table, while insureds die at the same rates used for the pure endowment portfolio; therefore, also in this case we have an implicit safety loading.\\
It is worth pointing out the higher expected demographic profit with respect to pure endowment portfolio and a different behaviour at time 1 between regular and single premium (see Figure 1, top). A higher profit is here realised when a single premium is charged because of greater values of the implicit safety loading\footnote{$\lambda_t$ is calculated as $\lambda_t=\dfrac{(q_t^*-q_t)}{q_t}$  and while in the Pure Endowment $\lambda_t$ was set equal to 15\%, using ISTAT2014 entails that when $t=0$, $\lambda_0=24.64\%$, then it increases until $t=19$ where $\lambda_{19}=32.34\%$}\\
Previous examples assumed a fully coincidence between technical rate and risk-free rates. The advantage is the fact that demographic profit is not affected by the behaviour of financial rates. We analyse now the effects of alternative risk-free rates and we focus only on the expected profit evaluated in a market-consistent framework.\\
Figure 2, left hand-side, reports the pattern of the expected profit for a pure endowment. In this case, same assumptions of Table 1 have been considered, but constant spot risk-free rates equal to $2\%$ are assumed. Higher risk-free rates lead to an increase of the initial profit due to a lower technical provision at the end of the year because of higher discounting effects\footnote{$\e{_1\tilde{y}_{1}^{MCV}}=\pi\cdot w_0 \cdot (1+j^*)-\e{\tilde{be}_{1}^{Rf(1),q}}\cdot w_0\cdot _1p_{40}$ This formulation is easily obtainable computing the expected value on formula (3), when t = 0}. A lower increase is observed in case of regular premiums because also future cash-in (premiums) are discounted at higher rates.
\begin{figure}[ht!]
	\centering
	\includegraphics[width=80mm]{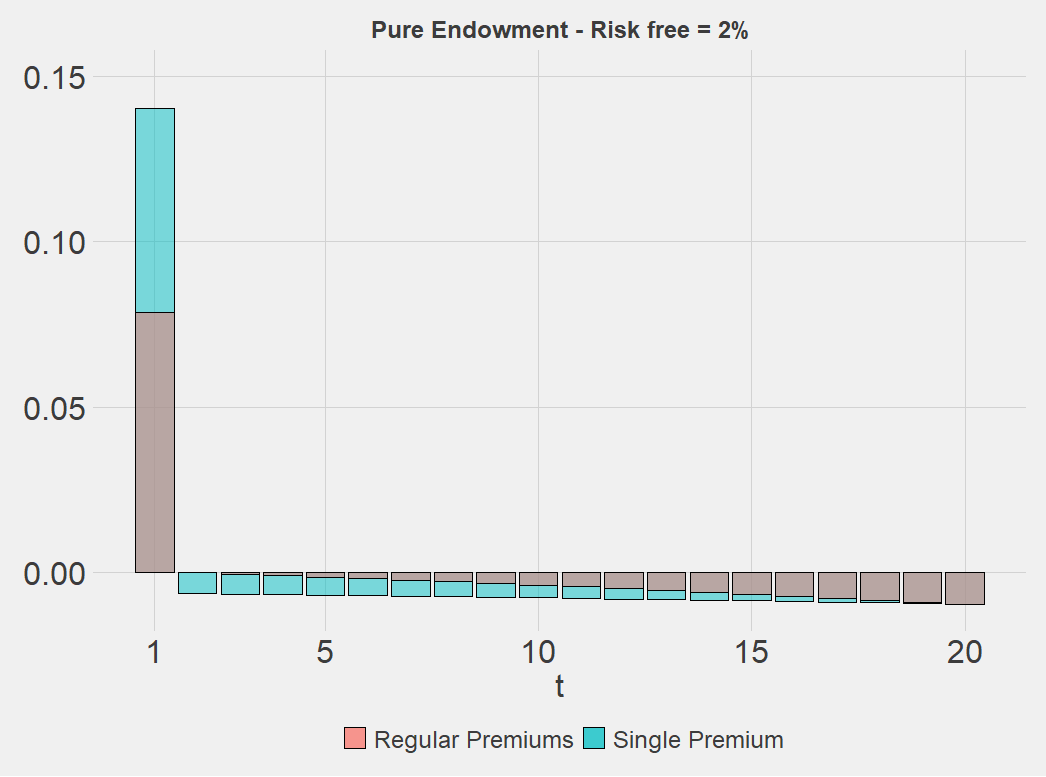}
	\caption{Expected demographic profit in the Pure Endowment \label{overflow}}
	\label{F2}
\end{figure}
\pagebreak
Following periods, namely for $t>1$, are instead characterized by expected losses. This behaviour can be explained by formula (\ref{2ndComponent}). As shown in Figure \ref*{3tassi}, the reserve jump $(\tilde{be}_t^{Rf(t),q}-epv_t^{j^*,q})$ is negative and is \deleted{capitalized}\added{accumulated} at a rate $j^*$. This amount is higher than the term $(\tilde{be}_{t+1}^{Rf(t+1),q}-epv_{t+1}^{j^*,q})$. The same result can be also explained by formula (\ref{EVtgreater1}), where, under the assumption that the best estimates at time $t$ and $t+1$ are positive, we have an expected loss if $j^*$ is smaller than the spot rate $i(t,t+1)$. 
\begin{figure}[ht!]
	\centering
	\includegraphics[width=80mm]{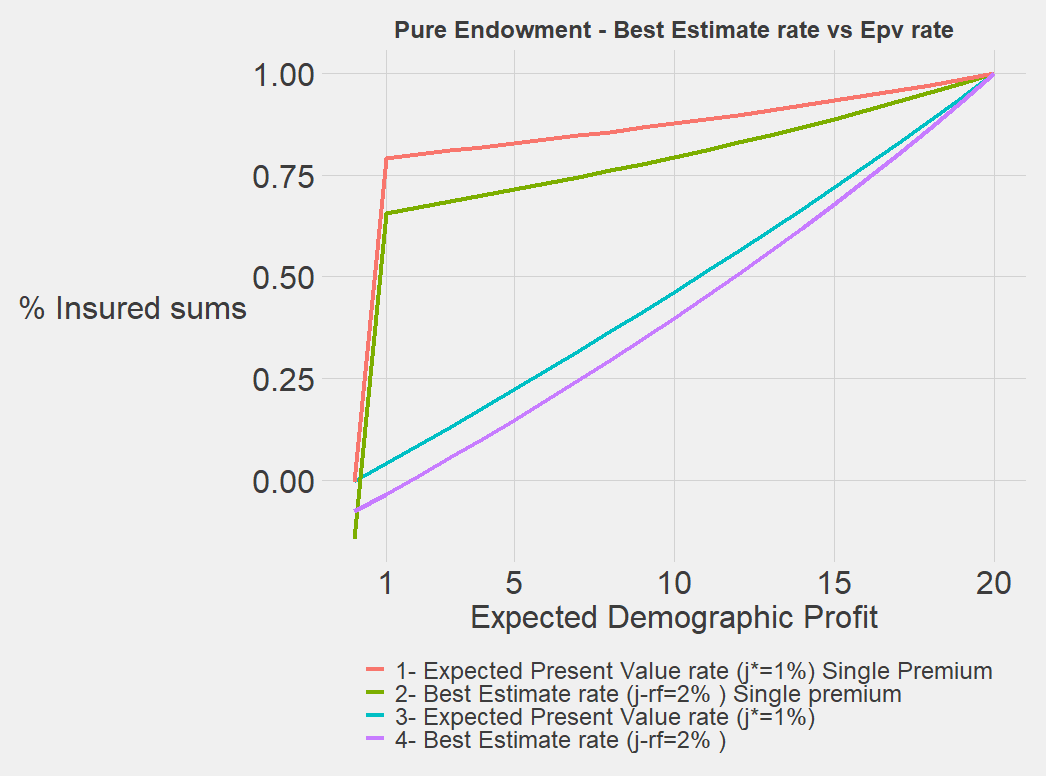}
	\caption{Expected present value and best estimate rates \label{overflow}}
	\label{3tassi}
\end{figure}
\\It is also interesting to note that in case of a negative best estimate we have an opposite result. In Figure \ref*{tassorfal20}, we consider again a pure endowment with the same characteristics but we assume \deleted{an extreme} \added{an extraordinarily large} flat rate equal to $20\%$. It could be noticed that, in line with formula (\ref{EVtgreater1}), in time periods in which $\tilde{be}_{t}^{Rf(t),q}$ is negative, a positive expected profit is observed (see $0\leq t\leq 11$ in Figure \ref*{tassorfal20}).
\begin{figure}[ht!]
	\centering
	\includegraphics[width=95mm]{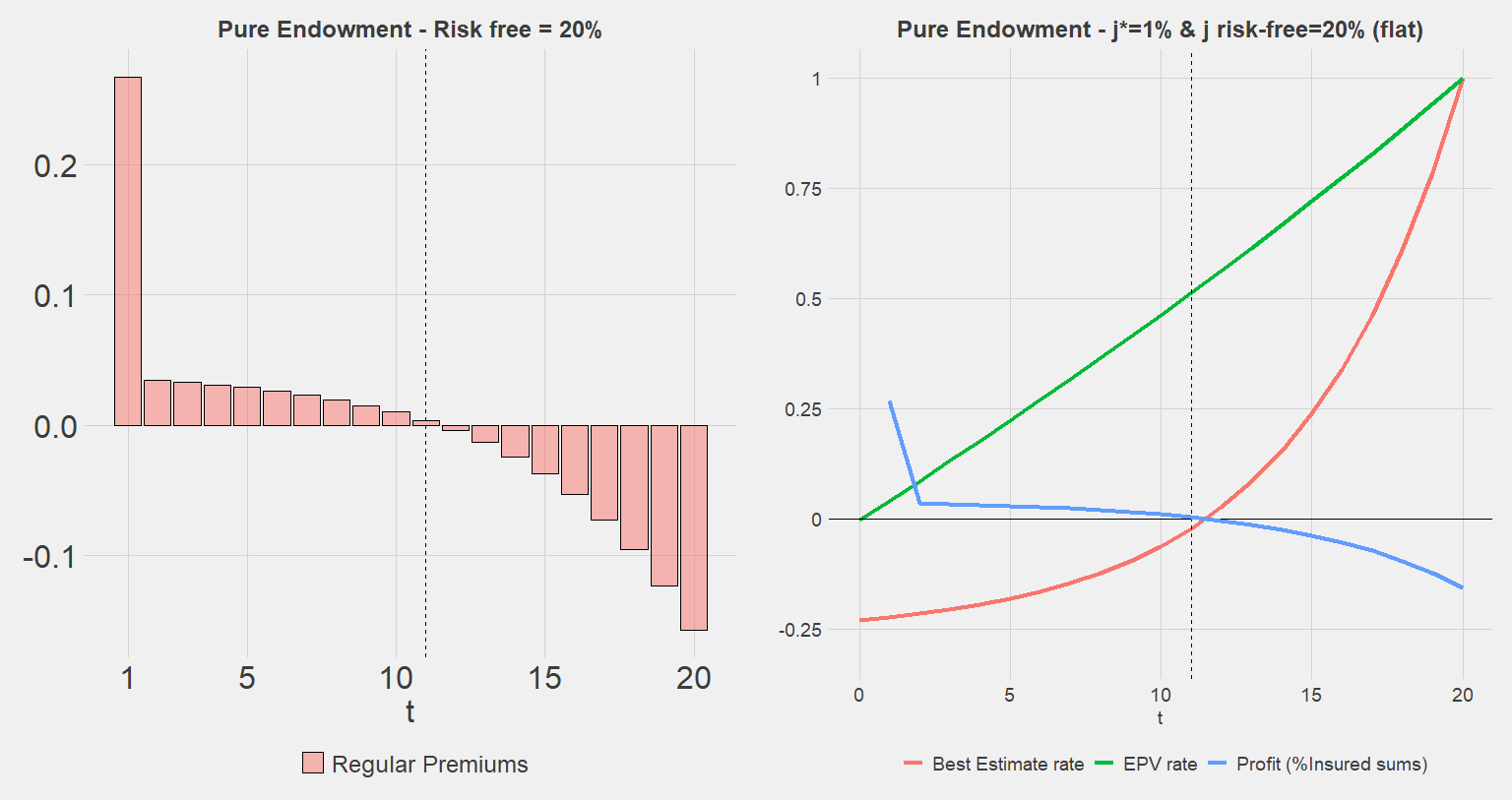}
	\caption{Expected demographic profit in the Term Insurance \label{overflow}}
	\label{tassorfal20}
\end{figure}
\\In Figure \ref*{5tcm} we report the behaviour of the expected profit in case of either a regular and a single premium for the term insurance. Policyholder and contract characteristics are the same reported in Table 1 and the risk-free rate is equal to $2\%$. Results confirm a consistent behaviour independent of the kind of policy considered.\\
However, because of the lower best estimate rate, the effect of financial rates is less relevant when a term insurance contract is considered.
\begin{figure}[ht!]
	\centering
	\includegraphics[width=110mm]{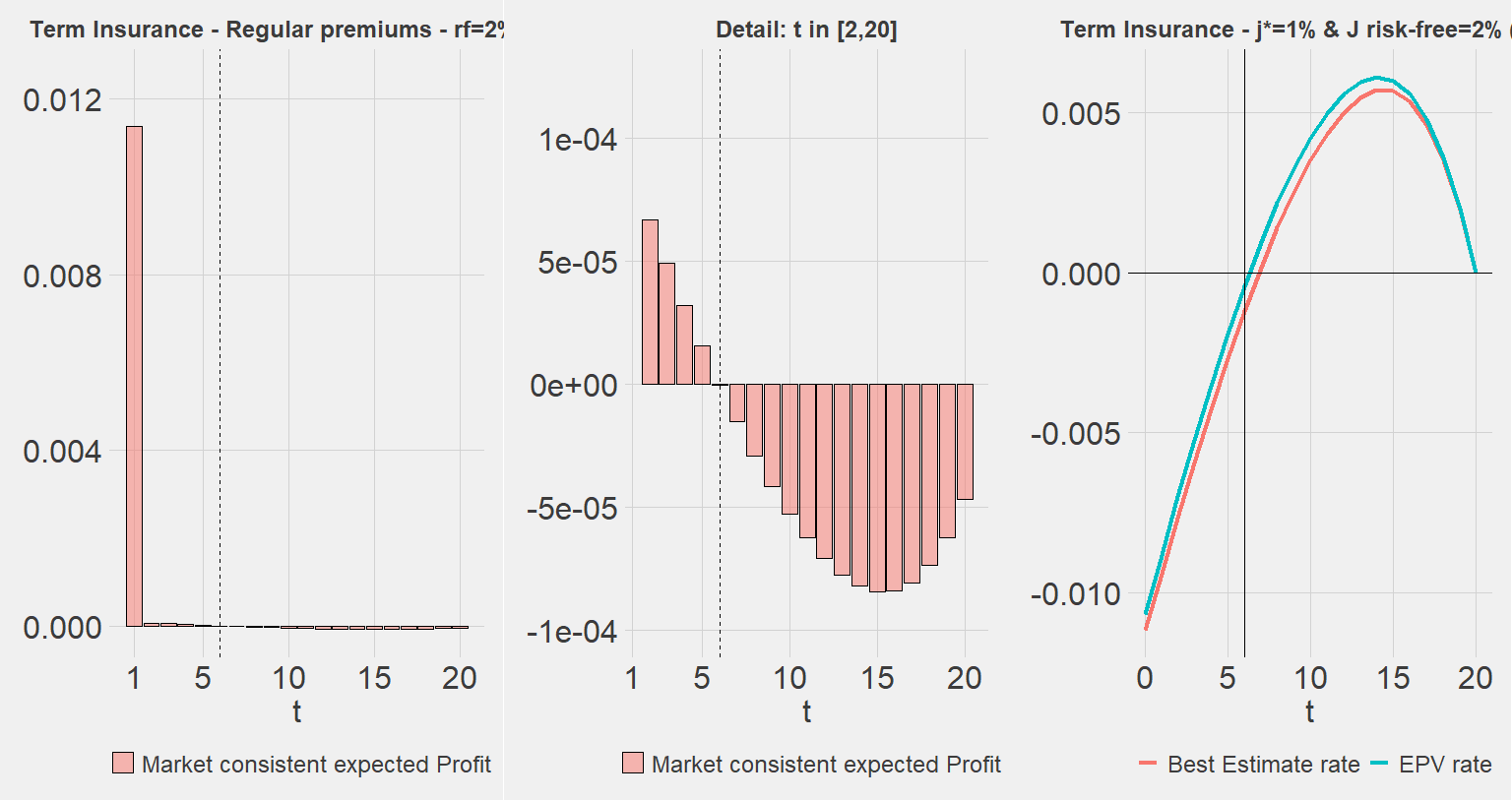}
	\caption{Expected demographic profit in the TI with high risk-free rates \label{overflow}}
	\label{5tcm}
\end{figure}
\\Finally, in this section, the analysis has been developed assuming a constant risk-free rate for all maturities. However, considering the case of a risk-free curve, similar comments follow from formula (\ref{2ndComponent}) and \ref{EVtgreater1} based on the comparison between the technical rate $j^*$ and the forward rates for each period $(t,t+1)$.\\Before presenting next Section, it is emphasized that for a particularly young insurance company, the strictly one year vision can lead to incorrect and partial results, indeed if the Solvency Capital Requirement were calculated over a three-year time horizon, the results would be opposite.

\section{The application of the model to non-participating life policies} 
\noindent In this section, we analyse the behaviour of our model in terms of volatility and capital requirement by developing a detailed case study. We consider a portfolio characterized by three \added{non-participating} insurance policies: a Pure Endowment, an Endowment and a Term Insurance. To this end, we are able to catch how the different characteristics of the contracts can affect the demographic profit distribution and the related capital requirement.\\
In Table 2, we summarize the general characteristics (age and contractual duration of the policy) and the expense loadings. For the sake of comparability, we assume that different policies have been underwritten by policyholders with the same characteristics. Furthermore, the \added{sums insured} at time 0, $w_0$, are calculated as the sum of the individual \added{sums insured} of each policyholder:
\begin{equation}
w_0=\sum_{i=1}^{l_0}C_i = l_0 \cdot \bar{C_0}
\label{capitaliassicurati}
\end{equation}
where $l_0$ is the number of policyholders in $t=0$.
\begin{table}[htb]
	\label{tab:t2}
	\centering
	\caption{Model parameters}
	\begin{tabular}{|l|r|}
		\hline
		Individual age at policy issue & 40             \\ \hline
		Policy duration                & 20           \\ \hline
		Premium type                & Annual premiums (20)             \\ \hline
		Acquisition loading              & 50\%            \\ \hline
		Collection loading                & 2.5\%            \\ \hline
		Management loading                & 0.15\%            \\ \hline
		Number of policyholders $(l_0)$       & 15,000         \\ \hline
		Expected value of the single insured sum        & 100,000         \\ \hline
		CV of the \added{sums insured}         & 1.99\%            \\ \hline
	\end{tabular}
\end{table}
The risk-free rate curve used is the one provided by EIOPA for the Euro area (at the end of 2015); we preferred to use a curve calibrated in a quiet period, which does not present particular types of stress. We have chosen the one without volatility adjustment. Demographic technical bases are summarized in Table 3 and are different for each policy because the goal is to highlight specific profit creation considering a realistic rating process. \\
Consequently, for the Endowment and the Term Insurance with positive sum-at-risk, the pricing is carried out with the ISTAT2014 demographic table. Pure Endowment, characterized by a negative sum-at-risk, has been priced multiplied the death probabilities, derived by ISTAT2016 life tables, by a coefficient $\alpha_t$ between $0.8$ and $0.9$ that depends on the age of the policyholder. Effective mortality of all portfolios is instead described by the ISTAT2016 table. \added{We are obviously aware that second-order mortality can be affected by self-selection or medical-selection (e.g. in term insurance products to identify calibrated risks). Similarly, term insurance and endowment policies are usually priced using a different life table. To assure a greater comparison between results we considered in this numerical analysis similar technical bases for different policies. The results can be easily adapted to the case of a higher customization of life tables.} \\
It is \added{also} noteworthy that future mortality rates can be obviously obtained using forecasting models (as, for instance, Lee-Carter model). However, main comments described in this section also hold in case of projected second-order life tables.
\begin{table}[]
	\centering
	\caption{Demgraphic assumptions }
	\label{tab:my-table}
	\begin{tabular}{|l|c|c|c|}
		\hline
		& \multicolumn{1}{l|}{Pure Endowment} & \multicolumn{1}{l|}{Endowment} & \multicolumn{1}{l|}{Term Insurance} \\ \hline
		First order - q* & $\alpha_t\cdot$ISTAT2016                          & ISTAT2014                      & ISTAT2014                           \\ \hline
		Second order - q & ISTAT2016                           & ISTAT2016                      & ISTAT2016                           \\ \hline
	\end{tabular}
\end{table}
The model is applied by means of Monte Carlo simulations. In particular, the number of deaths is simulated in a one-year time horizon by a Binomial distribution with parameters equal to the number of policyholders and the second-order death probability. To consider the variability of the \added{sums insured}, we extract the insured capital at the end of the year and the amounts paid in case of death by LogNormal distributions with mean and CV defined in Table 2.\\
Another key element is the assessment of the Best Estimate rate at the end of the period. It is assumed that, on average, the spot rates at the end of the year will coincide with the forward rates inferable from the risk-free curve of spot rates at the beginning of the year. The volatility of risk-free rates is introduced through the use of a Vašíček model (see \cite{vasicek1977equilibrium}); particularly effective in case of negative risk-free rates for shorter maturities. Hence, the calibration of the Vašíček model is carried out for each policy and at each time point, requiring that the expected value of the best estimate rate calculated at the end of the time span coincides with the best estimate rate calculated with the forward rates implicit in the spot rate curve at the beginning of the time horizon.\\
The results of the stochastic model are presented. Two points must be highlighted:
\begin{itemize}
	\item For each policyholder, 10 million simulations have been made. Therefore, the results are particularly consistent, especially in terms of volatility\footnote{With an Intel i7 8700K processor (working in parallel - 6 Cores, 12 Threads) it requires about 18 minutes}.
	\item The total amount of the \added{sums insured} at the inception of the policy ($t =0$) is equal to approximately 1.5 billion euros. It is therefore noted that any Capital Requirement, in terms of magnitude, must be compared with the value just mentioned, although at first glance it may seem particularly high
\end{itemize}
First of all, Table 4 shows\footnote{In all tables, T stands for Theoretical values, i.e. the exact characteristics of the random variables \added{computed using closed formulas}.} the results of the simulation model applied at the inception of the contract ($t =0$) and with a one-year view to the portfolio of Pure Endowment policies.
\begin{table}[H]
	\caption{Simulated MCV results - Pure Endowment}
	\centering
	\label{tab:my-table}
	\begin{tabular}{|l|r|r|r|}
		\hline
		Pure Endowment & t=0          & t=10       & t=19        \\ \hline
		$\e{_1\tilde{y}_{t+1}^{MCV}}(T)$& 150,339,562  & -8,882,965 & -19,893,883\\ \hline
		$\e{_1\tilde{y}_{t+1}^{MCV}}$& 150,349,827  & -8,760,096 & -19,894,122\\ \hline
		$\e{_1\tilde{y}_{t+1}^{MCV}}$ on $w_t$& 9.95\%      &-0.61\%            & -1.38\%            \\ \hline
		$\sigma(_1\tilde{y}_{t+1}^{MCV})$& 4,345,500    & 5,937,483  & 1,821,044           \\ \hline
		$\gamma(_1\tilde{y}_{t+1}^{MCV})$& -0.04        & -0.02      & 0.97           \\ \hline
		SCR& -138,942,272 & 24,208,313 & 23,333,458  \\ \hline
		SCR on $w_t$& -9.20\%      & 1.62\%     & 1.60\%      \\ \hline
	\end{tabular}
\end{table}
\begin{table}[H]
	\caption{Local GAAP results - Pure Endowment}
	\centering
	\label{tab:my-table}
	\begin{tabular}{|l|r|r|r|}
		\hline
		Pure Endowment & t=0          & t=10       & t=19        \\ \hline
		$\e{_1\tilde{y}_{t+1}^{LG}}(T)$& 2,528  & 230,554 & 1,040,705\\ \hline
		$\sigma(_1\tilde{y}_{t+1}^{LG})(T)$& 16,455    & 615,730  & 1,799,850           \\ \hline
		$\gamma(_1\tilde{y}_{t+1}^{LG})(T)$& 2.41        & 1.54      & 1.09           \\ \hline
		SCR & 19,463 & 792,091 & 2,397,872  \\ \hline
		SCR on $w_t$ & 0.01\% & 0.05\% & 0.16\%  \\ \hline
	\end{tabular}
\end{table}

\noindent It should be noted that the initial Best Estimate rate at the time $0^+$ is negative  (equal to -7.65\%). Despite the significant contribution of the implicit safety demographic loading, the largest share of expected profit derives from the difference between the first order financial rate ($j^*=1\%$) and the discounting spot rate for longer maturities. For instance, the expected survival benefits paid to the policyholder at the end of the coverage are discounted by using a spot rate equal to 1.57\%\footnote{Results are strongly influenced by the risk-free curve used. For instance, using the EIOPA curve for August 2020 we obtain an expected profit in $t=19$ in Table 4, because of a spot rate equal to 0.69\%.}.\\
Additionally, we have that the capital requirement, computed here using a value at risk at a 99.5\% confidence level, is negative. We have indeed that the huge expected profit allows to cover adverse fluctuation of demographic assumption also in the worst case \added{computed at the previously mentioned confidence level}.
A different situation is instead obtained assuming to be at time t=10 and t=19, respectively. In Table 4 we report also main characteristics of profit distribution as well as the capital requirement computed on a one-year view at different time periods.\\
It is interesting to note that after the first year of contract, where the expected profit \added{at inception} is accounted for, the risk-free forward rate higher than the technical rate leads to significant expected losses (in accordance to formula (\ref{EVtgreater1}). This effect increases as the year progressively grows because forward rates grow over time. Because of this behaviour we have positive requirements in both periods. \\
In Table 5, we provide results obtained applying the model in a local accounting framework. It is noteworthy that we observe that the standard deviation in $t=10$ increases due to the greater volatility of risk-free rates, while in $t=19$ only the strictly demographic volatility remains, because the only risk-free rate is known. As regards skewness, what happens is similar: in $t=0$ and $t=10$ the skewness deriving from the Vašíček model is the main driver, while in $t=19$ only the skewness of the purely demographic component remains. With reference to the simulated SCR, it should be noted that the previous Solvency 0 and Solvency I regulations indicated 0.3\% of positive sum-at-risk as a capital requirement for life underwriting risk without differentiations related to the characteristics of the insurance portfolio.\\
A similar analysis has been developed for an Endowment; we report in Table 6 main characteristics of the distribution of $_1\tilde{y}_{t+1}^{MCV}$ and the SCR ratio according to the three time periods; Table 7 summarizes analogous values computed in a local accounting framework.\\
\begin{table}[H]
	\caption{Simulated MCV results - Endowment}
	\centering
	\label{tab:my-table}
	\begin{tabular}{|l|r|r|r|}
		\hline
		Endowment & t=0          & t=10       & t=19        \\ \hline
		$\e{_1\tilde{y}_{t+1}^{MCV}}(T)$& 153,901,015  & -9,096,636 & -20,084,746 \\ \hline
		$\e{_1\tilde{y}_{t+1}^{MCV}}$& 153,905,829  & -9,138,027 & -20,084,746 \\ \hline
		$\e{_1\tilde{y}_{t+1}^{MCV}}$ on $w_t$& 10.19\%      &-0.61\%            & -1.38\%            \\ \hline
		$\sigma(_1\tilde{y}_{t+1}^{MCV})$& 4,588,717    & 6,020,325  & 0           \\ \hline
		$\gamma(_1\tilde{y}_{t+1}^{MCV})$& -0.05        & -0.02      & 0           \\ \hline
		SCR& -141,802,756 & 24,821,250 & 20,084,746  \\ \hline
		SCR on $w_t$& -9.39\%      & 1.66\%     & 1.38\%      \\ \hline
	\end{tabular}
\end{table}
\begin{table}[H]
	\caption{Local GAAP results - Endowment}
	\centering
	\label{tab:my-table}
	\begin{tabular}{|l|r|r|r|}
		\hline
		Endowment & t=0          & t=10       & t=19        \\ \hline
		$\e{_1\tilde{y}_{t+1}^{LG}}(T)$& 260,745  & 395,352 & 0\\ \hline
		$\sigma(_1\tilde{y}_{t+1}^{LG})(T)$& 750,933    & 595,986  & 0           \\ \hline
		$\gamma(_1\tilde{y}_{t+1}^{LG})(T)$& -2.41        & -1.54      & 0           \\ \hline
		SCR & 3,289,428 & 2,091,067 & 0  \\ \hline
		SCR on $w_t$ & 0.22\% & 0.14\% & 0\%  \\ \hline
	\end{tabular}
\end{table}
\noindent Endowment contract shows similar results to the pure endowment case in terms of both expected profit and capital requirement ratio. Main differences can be noticed in the last year of contract ($t=19$). We have indeed that, given the \added{fact that the payment of benefit is sure}, we have no volatility and hence, the capital requirement is only needed to face expected losses.\\
As well-known previous contracts are typically chosen for saving purposes and the financial profit is the key issue for an insurance company. Therefore, we investigate the behaviour of the demographic profit in case of a term insurance, that is typically chosen by policyholders for risk-protection purposes.
Main results are reported in Table 8.\\
\begin{table}[H]
	\caption{Simulated MCV results - Term Insurance}
	\centering
	\label{tab:my-table}
	\begin{tabular}{|l|r|r|r|}
		\hline
		Term Insurance & t=0          & t=10       & t=19        \\ \hline
		$\e{_1\tilde{y}_{t+1}^{MCV}}(T)$& 20,898,387  & -271,353 & -269,009 \\ \hline
		$\e{_1\tilde{y}_{t+1}^{MCV}}$& 20,891,192  & -269,927 & -269,285 \\ \hline
		$\e{_1\tilde{y}_{t+1}^{MCV}}$ on $w_t$& 1.38\%      &-0.02\%            & -0.02\%            \\ \hline
		$\sigma(_1\tilde{y}_{t+1}^{MCV})$& 777,012    & 1,224,059  & 1,820,823           \\ \hline
		$\gamma(_1\tilde{y}_{t+1}^{MCV})$& -2.42        & -1.48      & -0.97           \\ \hline
		SCR& -17,219,206 & 24,821,250 & 6,769,046  \\ \hline
		SCR on $w_t$& -1.14\%      & 0.36\%     & 0.47\%      \\ \hline
	\end{tabular}
\end{table}
\begin{table}[H]
	\caption{Local GAAP results - Term Insurance}
	\centering
	\label{tab:my-table}
	\begin{tabular}{|l|r|r|r|}
		\hline
		Term Insurance & t=0          & t=10       & t=19        \\ \hline
		$\e{_1\tilde{y}_{t+1}^{LG}}(T)$& 266,582  & 800,159 & 2,160,710\\ \hline
		$\sigma(_1\tilde{y}_{t+1}^{LG})(T)$& 767,745    & 1,206,226  & 1,799,850           \\ \hline
		$\gamma(_1\tilde{y}_{t+1}^{LG})(T)$& -2.41        & -1.54      & -1.09           \\ \hline
		SCR & 3,335,518 & 4,265,074 & 4,320,054  \\ \hline
		SCR on $w_t$ & 0.22\% & 0.29\% & 0.29\%  \\ \hline
	\end{tabular}
\end{table}
\noindent It is interesting to note that the expected gain accounted for in $t=1$, although strictly greater than the expected losses of the subsequent periods, is lower than other policies ones (\added{e.g. pure endowment, endowment}). In this case, the very small volume of mathematical reserves ensures that expected losses and expected profits are very low. \added{The previous comment is also explained by the fact that we are assuming the same first and second-order life tables in term insurance and endowment.} On the other hand, despite the volatility of sums insured paid in case of death, a lower ratio between the capital requirement and the sums insured is observed.\\
Comparing the results of the Term \added{Insurance} with those obtained in a Local GAAP context (see Table 9), we observe that also in this case the volatility in $t=0$ and $t=10$ is mainly driven by the volatility of the Vašíček model, while in $t=19$ it remains only the volatility of the strictly demographic component. As in the case of the Pure Endowment, it is clearly observed that where the financial component is zero (in $t=19$ the only spot rate is known), only the purely demographic skewness remains.\\
Finally, it should be noted that, despite the different nature of the alternative policies, a similar effect of the implicit forward rates is noticed both in terms of expected gains / losses, and in terms of capital requirement
\section{Conclusions} 
\noindent In this paper, we focus on the evaluation of the capital requirements for both mortality and longevity risk.  In particular, we adapt classical actuarial relations to the market consistent framework required by Solvency II directive. We provide a specific model able to catch the characteristics of demographic rik for \added{non-participating} life insurance contracts. As well-known, in a local accounting context, differences between expected and observed mortality rates is the key topic for assessing demographic risk. We show that in a market-consistent framework, the financial component cannot be completely separated from the purely demographic one.  We prove indeed that in case second-order demographic assumptions are stable over time, the connection between the financial guaranteed rate and the risk-free rate curve becomes the key element for the assessment for one-year demographic risk. Hence, to have a complete view of the insurance position, main results must be then compared with the financial profit and the capital requirement for market risk too. Therefore, further research should regard an integrated assessment of both demographic and financial risks that could be helpful for defining strategies and future management actions regarding both portfolio characteristics and asset allocation.
\added{Finally, further developments will exploit the proposed model and closed formulas to quantify both idiosyncratic and systematic volatility to quantify the Solvency Capital Requirement.}
\newpage

\bibliographystyle{acm}
\bibliography{BibMCV}

\newpage

\appendix
\section{Homans' revised decomposition}
\label{sec:app1}
\noindent We show here the decomposition of formula (\ref{eq:tpmcv}) in five components. 
	For the sake of brevity, we limit this Appendix to the definition of the five components we found. The proof follows by simple algebra.
\begin{equation}
	\tilde{Y}_{t+1}^{MCV}=_1\tilde{y}_{t+1}^{MCV}+_2\tilde{y}_{t+1}^{MCV}+_3\tilde{y}_{t+1}^{MCV}+_4\tilde{y}_{t+1}^{MCV}+_5\tilde{y}_{t+1}^{MCV}
	\label{eq:5components}
\end{equation}
It is noteworthy that we focus exclusively on without-profit policies. The reason for this choice lies in the fact that the model aims at isolating the demographic component. \\
We give now the definition of the five components expressed in rate notation.
The first component represents the \added{demographic} profit and it's defined as in equation \ref{eq:demographicprofit}: 
\begin{equation}
\begin{aligned}
_1\tilde{y}_{t+1}^{MCV} = \ &[\tilde{be}_t^{Rf(t), q(t)}+b_{t+1}(1-\alpha^*-\beta^*)-\gamma^*] (\tilde{w}_t-\tilde{s}_{t+1})(1+j^*)+\\
&-(\tilde{x}_{t+1}+\tilde{be}_{t+1}^{Rf(t+1), q(t+1)})
\end{aligned}
\label{eq:demographicprofit2}
\end{equation}
While the second profit component as defined in equation (\ref{eq:financialprofit}) could be seen as:
\begin{equation}
	\begin{aligned}
		_2\tilde{y}_{t+1} =&(\tilde{j}_{t+1}-j^*)\cdot(\tilde{be}_t^{Rf(t), q(t)}\cdot\tilde{w}_t + b_{t+1}(1-\alpha^*-\beta^*)(\tilde{w}_t-\tilde{s}_{t+1})+\\
		& -(\gamma^*\tilde{w}_t)-(g^*_t\cdot\tilde{be}_t^{Rf(t),q(t)}\cdot\tilde{s}_{t+1})
	\end{aligned}
	\label{eq:financialprofit2}
\end{equation}
The lapse profit, is defined as:
\begin{equation}
_3\tilde{y}_{t+1}^{MCV} = \ (\tilde{be}_t^{Rf(t), q(t)} -\gamma^*-g^*_t\cdot\tilde{be}_t^{Rf(t), q(t)})\cdot(1+j^*)\cdot\tilde{s}_{t+1}
\label{eq:lapseprofit}
\end{equation}
Where $g^*_t$  is a penalization coefficient that considers a surrender penalty:
\begin{equation}
g^*_t = 
\begin{cases}
0&\ if\ t<\tau  \\  (1+j_s)^{-(m-t)}&\ if\ t>=\tau 
\end{cases}
\label{eq:zillmercoeff}
\end{equation}
where $\tau>0$ and $j_s^*<j^*$  are fixed by the undertaking. \\
The fourth component of profit is the expense one and it is defined in the following way:
\begin{equation}
_4\tilde{y}_{t+1} = (1+j^*)[(\Delta\alpha^*_{t+1}+\Delta\beta^*_{t+1})\cdot b_{t+1}\cdot(\tilde{w}_t-\tilde{s}_{t+1})+\Delta\gamma^*_{t+1}\cdot\tilde{w}_t]
\label{eq:expenseprofit}
\end{equation}
where $\Delta\alpha_{t+1}^*$, $\Delta\beta_{t+1}^*$ and $\Delta\gamma_{t+1}^*$ depend on the differences between the first order expense assumptions and the realistic ones.\\
The last component, is the residual profit:
\begin{equation}
\begin{aligned}
_5\tilde{y}_{t+1} =\  &(\tilde{j}_{t+1}-j^*)[(\Delta\alpha^*_{t+1}+\Delta\beta^*_{t+1})\cdot b_{t+1}\cdot(\tilde{w}_t-\tilde{s}_{t+1})+\\
&+\Delta\gamma^*_{t+1}\cdot\tilde{w}_t]
\end{aligned}
\label{eq:5profit}
\end{equation}
\label{homans}

	

\end{document}